\documentclass[letterpaper, 10pt, onecolumn,conference]{ieeeconf}  

\IEEEoverridecommandlockouts                              

\usepackage{amsmath}
\usepackage{amssymb}
\usepackage{amsfonts}
\usepackage{mathptmx} 
\usepackage{times}
\usepackage{graphicx}
\DeclareMathAlphabet{\bit}{OML}{cmm}{b}{it}
\newtheorem{lem}{Lemma}
\newtheorem{thm}{Theorem}

\usepackage{datetime}
\usepackage{times} 
\usepackage{framed} 
\usepackage{graphicx} 
\usepackage{amsmath} 
\usepackage{amssymb}  

\def\<{\leqslant}           
\def\>{\geqslant}           

\def\d{\partial}
\def\wh{\widehat}
\def\wt{\widetilde}

\def\Re{\mathrm{Re}}   

\def\mR{\mathbb{R}}    
\def\mC{\mathbb{C}}    

\def\Tr{\mathrm{Tr}}       
\def\rT{\mathrm{T}}        

\def\bE{\mathbf{E}}    

\def\[[[{[\![\![}
\def\]]]{]\!]\!]}

\def\bra{{\langle}}
\def\ket{{\rangle}}

\def\re{\mathrm{e}}        
\def\rd{\mathrm{d}}        

\def\x{\times}
\def\ox{\otimes}

\def\cF{\mathcal{F}}

\def\cP{{\mathcal P}}

\def\cov{\mathbf{cov}}

\def\cN{\mathcal{N}}

\def\mS{\mathbb{S}}

\def\eps{\epsilon}
\def\Ups{\Upsilon}

\title{\LARGE \bf Dissipative Linear Stochastic Hamiltonian Systems%
 $^*$}

\usepackage{datetime}

\author{Igor G. Vladimirov$^{\dagger}$, \quad Ian R. Petersen$^{\dagger}$%
\thanks{$^*$This work is supported by the Australian Research Council under grant DP160101121.}
\thanks{$^\dagger$Research School of Engineering, College of Engineering and Computer Science, Australian National University, ACT 2601, Canberra, Australia.
{\tt igor.g.vladimirov@gmail.com}, {\tt i.r.petersen@gmail.com}.
}
}

\begin{document}
\maketitle
\thispagestyle{empty}
\begin{abstract}
This paper is concerned with stochastic Hamiltonian systems which model a class of open dynamical systems subject to random external forces. Their dynamics are governed by Ito stochastic differential equations whose structure is specified by
a Hamiltonian, viscous damping parameters  and system-environment coupling functions. We  consider energy balance relations for such systems with an emphasis on linear stochastic Hamiltonian (LSH)  systems with quadratic Hamiltonians and linear coupling. For LSH systems, we also discuss stability conditions, the structure of the invariant measure and its relation with stochastic versions of the virial theorem. Using Lyapunov functions, organised as deformed Hamiltonians, dissipation relations are also considered for LSH systems driven by statistically uncertain external forces. An application of these results to feedback connections of LSH systems is outlined.
%
\end{abstract}



\section{\bf Introduction}

The internal dynamics of physical systems and their interaction with the surroundings are  strongly influenced by conservation laws. This includes the energy balance relations which
manifest themselves in the structure of governing equations for such systems.
Being woven into Lagrangian and Hamiltonian mechanics, this specific structure is taken into account and plays an important part in analysis and control design for
port-Hamiltonian systems \cite{OVMM_2001,OVME_2002,VJ_2014}. Energy transfer and dissipation  (for example, in the form of viscous damping, mechanical friction or electrical resistance) are crucial mechanisms which underlie the collective behavior in connections of such systems and can be used for achieving their stability and other performance specifications \cite{V_2016,W_1972}. The control-by-interconnection paradigm is alternative to the measurement-actuation control approach and applies not only to classical systems but also extends to coherent (that is, measurement-free) quantum control settings \cite{JG_2010,P_2014}.
Despite the novelty of its modern applications, this approach was used in centrifugal governors \cite{M_1868} long before the age of controllers with analog or digital processing of electrical signals.

The energy flows between interacting systems of interest (for example, the plant and controller) obey
balance equations which are
usually formulated in terms of the internal energy of the system, specified by its Hamiltonian, the energy dissipation and  the system-environment coupling. These concepts have a general representation in the form of storage and  supply rate functions in the theory of dissipative systems \cite{W_1972}, where the underlying processes (including the external forces) are usually represented by deterministic functions of time (satisfying local square integrability conditions).
%
%
For finite-dimensional linear time-invariant systems,  the properties of being passive, positive real or negative imaginary (in the case of position variables instead of the velocity as the output) admit criteria in the form of linear matrix inequalities for the transfer functions in the frequency domain or the state-space matrices themselves \cite{PL_2010,XPL_2010}. These criteria have extensions to infinite networks of such systems \cite{VP_2018}.



The present paper is concerned with a class of stochastic Hamiltonian systems driven by random external forces. Their dynamics are governed by Ito stochastic differential equations (SDEs) whose structure is specified by system energetics in terms of a Hamiltonian, viscous damping parameters  and system-environment coupling functions. We  consider energy balance relations for such systems with an emphasis on linear stochastic Hamiltonian (LSH)  systems with quadratic Hamiltonians and linear coupling. For LSH systems, we also discuss stability conditions, the structure of the invariant Gaussian measure and its relation with stochastic versions of the virial theorem \cite{FBKB_2016}. Using Lyapunov functions, organised as deformed Hamiltonians, robust stability estimates are also obtained for LSH systems driven by statistically uncertain external forces. The latter are modelled as Ito processes whose drift and diffusion parts satisfy sector boundedness conditions which are similar to those in \cite{PUS_2000}. We outline an application of these results to
feedback connections of LSH systems.

The paper is organised as follows.
Section~\ref{sec:sys} specifies the class of stochastic Hamiltonian systems.
Section~\ref{sec:bal} discusses energy balance relations. 
Section~\ref{sec:lin} describes the class of LSH systems and provides stability conditions.
Section~\ref{sec:inv} studies the invariant measure of the LSH system and its connection with the virial  theorem.
Section~\ref{sec:diss} discusses dissipation relations for the LSH system driven by a statistically uncertain random force.
Section~\ref{sec:loop} applies these results to feedback connections of LSH systems.
Section~\ref{sec:conc} provides concluding remarks.

\section{\bf Stochastic Hamiltonian systems}
\label{sec:sys}

 We consider a class of stochastic Hamiltonian systems with $n$ degrees of freedom in the phase space $\mR^{2n} = \mR^n \x \mR^n$, which is the product of the position and momentum spaces. The position of the system is specified by a vector $q:= (q_k)_{1\< k\< n} \in \mR^n$ of generalised coordinates, with its time derivative $\dot{q} = (\dot{q}_k)_{1\< k\< n} \in \mR^n$ being the generalised velocity vector. The kinetic energy of the system is a position-dependent quadratic form
 \begin{equation}
 \label{kin}
    T(q,p):=
    \frac{1}{2}
    \|\dot{q}\|_{M(q)}^2
    =
    \frac{1}{2} \|p\|_{M(q)^{-1}}^2
 \end{equation}
of the velocity or the corresponding momentum vector
\begin{equation}
\label{pqdot}
    p:=
    \frac{1}{2}
    \d_{\dot{q}} (\|\dot{q}\|_{M(q)}^2)
    =
    M(q)\dot{q},
\end{equation}
where $\|v\|_N:= \sqrt{v^{\rT}N v} = |\sqrt{N}v|$ is a weighted Euclidean semi-norm of a real vector $v$ specified by a positive semi-definite matrix $N$. Here,
$M(q)$ is a real positive definite symmetric mass matrix (whose role, in the case of rotational degrees of freedom, can also be played by the tensor of inertia which may change together with the system configuration over the course of its movement).  The total energy of the system is quantified by a Hamiltonian  $H: \mR^{2n} \to \mR$ on the phase space as the sum of the kinetic energy (\ref{kin}) and
the potential energy $V: \mR^n \to \mR$:
\begin{equation}
\label{H}
    H(q,p)
    :=
    T(q,p)
    + V(q),
    \qquad
    q,
    p \in \mR^n.
\end{equation}
Both functions $V$ and $M$ are assumed to be twice continuously differentiable, so that $H$ inherits this property. The position $q$ and the momentum $p$ of the system  at time $t\>0$  evolve  according to the equations
\begin{align}
\label{SH1}
    \dot{q}
    &= \d_pH(q,p) =  M(q)^{-1}p,\\
\label{SH2}
    \rd p
    & =
    -(\d_q H(q,p) + F(q)\dot{q} )\rd t
    +
    G(q)\rd W,
\end{align}
the first of which is an ODE following from (\ref{pqdot}), while the second one is an Ito SDE \cite{KS_1991,LS_2001} driven by an $\mR^m$-valued  Ito process $W:=(W_k)_{1\< k\< m}$. The latter models an external random force acting on the system and will be
specified in the next section. The map $G:\mR^n\to \mR^{n\x m}$ describes the \emph{dispersion matrix} of the SDE (\ref{SH2}), while $F:\mR^n\to \mS_n^+$ specifies the Langevin viscous damping force $-F(q)\dot{q}$ (with $\mS_n^+$ the set of real positive semi-definite symmetric matrices  of order $n$).
%
%
In accordance with its physical meaning, the rate of work (per unit time) of the damping force  over the system is non-positive:
\begin{equation}
\label{damp}
    -\dot{q}^{\rT} F(q)\dot{q}
    =
%
    -\|\dot{q}\|_{F(q)}^2
    \< 0.
\end{equation}
The $\mR^{2n}$-valued state process
\begin{equation}
\label{XQP}
    x
    :=
    {\begin{bmatrix}
        q\\
        p
    \end{bmatrix}}
\end{equation}
satisfies the SDE
\begin{equation}
\label{dX}
    \rd x
    =
    \left(
        J
        -
        {\begin{bmatrix}
          0 & 0 \\
          0 & 1
        \end{bmatrix}}
        \ox F
    \right)
    H'
        \rd  t
        +
        {\begin{bmatrix}
          0 \\
          1
        \end{bmatrix}}
        \ox
        G\rd W,
\end{equation}
which is obtained by combining (\ref{SH1}) with (\ref{SH2}) (with
the arguments omitted for brevity).
Here, $\ox$ is the Kronecker product of matrices,  and
\begin{equation}
\label{J}
    J:=
    {\begin{bmatrix}
        0 & 1 \\
        -1 & 0 \\
    \end{bmatrix}}
    \ox I_n
\end{equation}
is the symplectic structure matrix
which generates the Poisson bracket  \cite{A_1989}
\begin{equation}
\label{Poiss}
    \{\varphi, \psi\}
    :=
    \varphi'^{\rT} J \psi'
    =
    \d_q \varphi^{\rT}\d_p \psi - \d_p \varphi^{\rT}\d_q \psi
\end{equation}
for differentiable real-valued functions 
on the phase space.
Also,
$(\cdot)'$ denotes the gradient of a function with respect to all its variables, so that
\begin{equation}
\label{H'}
    H'
    =
    {\begin{bmatrix}
        \d_q H \\
        \d_p H
    \end{bmatrix}}
\end{equation}
consists of the gradients of the Hamiltonian over the positions and momenta given by
\begin{equation}
\label{dHdq}
    \d_q H
    =
    V'(q)
    -
    \frac{1}{2}
    \big(
        p^{\rT}M^{-1} (\d_{q_k}M) M^{-1}p
    \big)_{1\< k\< n}
\end{equation}
and (\ref{SH1}).  While $-V'$ describes the potential force field, the additional ``centrifugal'' terms of $\d_qH$ in (\ref{dHdq}), which depend on the velocity in a quadratic fashion,  come from the dependence of $M$ on $q$ and the identity $\d_{q_k}(M^{-1}) =-M^{-1} (\d_{q_k}M) M^{-1}$.

\section{\bf Energy balance relations}\label{sec:bal}

Throughout this section, it is assumed that the stochastic differential of the Ito process $W$, which drives the SDE (\ref{SH2}), is given by
\begin{equation}
\label{dW}
  \rd W(t) = \alpha(t) \rd t + \beta(t) \rd \omega(t).
\end{equation}
Here, $\omega$ is a standard Wiener process \cite{KS_1991} in $\mR^m$ with respect to a filtration $\cF:= (\cF_t)_{t\> 0}$. Also, $\alpha$ and $\beta$ are $\cF$-adapted random processes with values in $\mR^m$ and $\mR^{m\x m}$,  respectively, satisfying
\begin{equation}
\label{good}
    \int_0^t
    (|\alpha(\tau)| + \|\beta(\tau)\|^2)
    \rd \tau
    <
    +\infty
\end{equation}
almost surely
for any time $t>0$, where $\|K\|:= \sqrt{\bra K,K \ket}$ is the Frobenius norm associated with the inner product $\bra K,L\ket:= \Tr(K^\rT L)$ of real matrices 
 (the validity of (\ref{good}) does not depend on a particular choice of the matrix norm).
 Due to the properties of the standard Wiener process $\omega$, the fulfillment of (\ref{good}) leads to a finite quadratic variation \cite{KS_1991}
\begin{equation}
\label{Wquad}
    [W]_t
    =
    \int_0^t
    \Tr \Sigma(\tau)
    \rd \tau
\end{equation}
for the process $W$ in (\ref{dW}) over the time interval $[0,t]$ represented
in terms of
the diffusion matrix 
\begin{equation}
\label{diffmat}
  \Sigma(t) := \beta(t)\beta(t)^\rT.
\end{equation}
A combination of the Ito lemma \cite{KS_1991} with (\ref{SH1}), (\ref{SH2}), (\ref{H'}), (\ref{dW}) and (\ref{diffmat}) leads to the stochastic differential of the Hamiltonian  in (\ref{H}):
\begin{align}
\nonumber
    \rd H
    &=
    H'^{\rT}\rd x
    +
    \frac{1}{2}
    \bra
        \d_p^2 H,
        G\Sigma G^\rT
    \ket
    \rd t\\
\nonumber
    &=
    \Big(
        \{H,H\}
        -
        \d_p H^\rT F\dot{q}
        +
                \frac{1}{2}
        \bra
            M^{-1},
            G\Sigma G^\rT
        \ket
    \Big)\rd t
        +
        \d_p H^{\rT}
        G \rd W\\
\label{dH}
        &=
        \Big(
            -
            \|\dot{q}\|_F^2
    +
            \frac{1}{2}
            \bra
                G^\rT M^{-1}G,
                \Sigma
            \ket
        \Big)\rd t
        +
        \dot{q}^\rT
        G \rd W,
\end{align}
where
$
    \d_p^2H = M^{-1}
$
is the Hessian of $H$ with respect to $p$ in view of (\ref{kin}) and (\ref{H}). Here, use is also made of the property  $\{H,H\} =0$ for the Poisson bracket (\ref{Poiss}), whereby the potential energy $V$ does not enter the right-hand side of (\ref{dH}).   The term $-\|\dot{q}\|_F^2$ in the drift of (\ref{dH}) is the dissipation rate (\ref{damp}), and the diffusion part $\dot{q}^\rT G \rd W$ is the incremental work of the external force on the system. The stochastic nature of the setting under consideration manifests itself in the term $\frac{1}{2} \bra G^\rT M^{-1}G, \Sigma \ket\>0$, which  comes from the diffusion part of the Ito process $W$ in (\ref{dW}). This additional term reflects the specific features  of the Ito stochastic calculus and is absent in the Stratonovich formulation   of SDEs \cite{RW_2000}.
In the case $F=0$ and $G=0$ (when there is no damping and the system is not affected by the environment), the relation (\ref{dH}) reduces to the ODE $\dot{H} = 0$, which describes the conservation of energy in closed Hamiltonian systems. The diffusion term $G\rd W$ of the SDE  (\ref{SH2}) is organised as a linear combination of potential force fields (with random ``coefficients'' $\rd W_1, \ldots, \rd W_m$)  if the dispersion matrix $G$ is given by
\begin{equation}
\label{gh}
  G(q) = L'(q)^{\rT},
  \qquad
  q\in \mR^n.
\end{equation}
Here, $L':= (\d_{q_k}L_j)_{1\< j\< m,1\< k\< n}\in \mR^{m\x n}$ is the Jacobian matrix for a map $L:=(L_j)_{1\< j\< m}: \mR^n\to \mR^m$
formed
from system-environment coupling functions $L_1, \ldots, L_m$ which are assumed to be continuously differentiable (see also \cite{V_2016} and references therein).  In comparison with $V$, each of the functions $L_k$ plays the role of the negative of potential energy, which generates the corresponding column $L_k'= (\d_{q_j} L_k)_{1\< j\< n}$ of the matrix $G$ in (\ref{gh}). Therefore, the incremental work  of the external force on the system in (\ref{dH}) can be represented as
\begin{equation}
\label{pot}
  \dot{q}^\rT
        G \rd W
        =
          (L'\dot{q})^\rT
         \rd W
        =
        \dot{y}^{\rT}\rd W,
\end{equation}
where use is made of the time derivative of the composite function
\begin{equation}
\label{y}
  y := L(q),
\end{equation}
which will be interpreted as an $\mR^m$-valued output of the stochastic Hamiltonian system (\ref{dX}). As an input-output operator, the resulting system is specified by the quadruple $(V,M,F,L)$ of the potential energy function $V$, the mass and damping matrices $M$ and $F$, and the vector $L$ of coupling functions. The structure of the output $y$ in (\ref{y}) as a function of the position  variables (rather than the velocity $\dot{q}$ in (\ref{pot})) suggests an analogy with negative-imaginary linear systems \cite{LP_2008,PL_2010,XPL_2010}, whose relation with positive real systems involves a ``rotation'' of the transfer functions 
by $\frac{\pi}{2}$.

\section{\bf Linear stochastic Hamiltonian systems}\label{sec:lin}

Consider the case when the mass matrix $M\succ 0 $ in (\ref{H}) is constant and the potential energy $V$ is a quadratic function of the position vector $q$:
\begin{equation}
\label{VK}
    V(q)
    =
    \frac{1}{2}
    q^\rT K q,
\end{equation}
where $K$ is a real symmetric \emph{stiffness matrix} of order $n$. Then the Hamiltonian $H$ is a quadratic form of the state vector $x$ from (\ref{XQP}):
\begin{equation}
\label{HR}
  H
  =
  \frac{1}{2}
  x^\rT R x,
\end{equation}
which is specified by the \emph{energy matrix}
\begin{equation}
\label{R}
  R
  :=
  {\begin{bmatrix}
    K & 0\\
    0 & M^{-1}
  \end{bmatrix}}.
\end{equation}
Also, suppose the system-environment coupling functions $L_1, \ldots, L_m$ are linear:
\begin{equation}
\label{hN}
  L(q) = N q,
\end{equation}
where $N \in \mR^{m\x n}$ is a \emph{coupling matrix}. Furthermore, let the damping matrix $F$ in (\ref{SH2}) be constant.
In view of (\ref{VK})--(\ref{hN}), the corresponding stochastic Hamiltonian system, described by  (\ref{SH1}), (\ref{SH2}) and (\ref{y}), is linear:
\begin{align}
\label{qdotlin}
    \dot{q}
    & =
    M^{-1} p,\\
\label{dplin}
    \rd p
    & =
    -(Kq + FM^{-1}p)\rd t
        +
          N^\rT
        \rd W,\\
\label{ylin}
    y &
    =
      Nq.
\end{align}
In accordance with (\ref{dX}), an equivalent form of (\ref{qdotlin})--(\ref{ylin}) in terms of the state vector $x$ is given by
\begin{align}
\label{dXlin}
    \rd x
    & =
    A x \rd t + B\rd W,\\
\label{ylin2}
    y &
    =
    Cx,
\end{align}
where the state-space matrices $A \in \mR^{2n \x 2n}$, $B \in \mR^{2n \x m}$, $C\in \mR^{m\x 2n}$ are computed as
\begin{align}
\label{A}
    A
    & :=
        \left(
        J
        -
        {\begin{bmatrix}
          0 & 0 \\
          0 & 1
        \end{bmatrix}}
        \ox
        F
    \right)
    R
     =
        {\begin{bmatrix}
          0 & M^{-1} \\
          -K & -FM^{-1}
        \end{bmatrix}},\\
\label{B}
        B
        & :=
        {\begin{bmatrix}
          0 \\
          N^{\rT}
        \end{bmatrix}},\\
\label{C}
        C
        & :=
        {\begin{bmatrix}
          N & 0
        \end{bmatrix}},
\end{align}
with $J$ given by  (\ref{J}),
cf. \cite[Eq. (20)]{V_2016}.
The linear stochastic Hamiltonian (LSH) system, described by (\ref{qdotlin})--(\ref{ylin})  (or (\ref{dXlin})--(\ref{C})),  is specified by the quadruple $(K,M,F,N)$ of the stiffness, mass, damping and coupling matrices, respectively. The fact that the matrices $A$, $B$, $C$ in (\ref{A})--(\ref{C}) have a special structure due to energetics of the LSH system is reminiscent of the nature of physical realizability conditions for linear quantum stochastic systems \cite{JNP_2008,
SP_2012}. The following lemma shows that the mass matrix can be considered the identity matrix.

\begin{lem}
\label{lem:trans}
As an input-output operator (from $W$ to $y$), the LSH system $(K,M,F,N)$ is equivalent to $(\wt{K},I_n,\wt{F},\wt{N})$ with the identity mass matrix and
the following stiffness, damping and coupling matrices:
\begin{align}
\label{Knew}
    \wt{K}
    & :=
    M^{-1/2}K M^{-1/2},\\
\label{Fnew}
    \wt{F}
    & :=
    M^{-1/2}F M^{-1/2},\\
\label{Nnew}
    \wt{N}
    & :=
    N M^{-1/2},
\end{align}
where $M^{-1/2}:= \sqrt{M^{-1}}$.
\hfill$\square$
\end{lem}
\begin{proof}
The equivalence is established by representing (\ref{qdotlin})--(\ref{ylin}) as
\begin{align*}
    \dot{\wt{q}}
    & =
    \wt{p},\\
    \rd \wt{p}
    & =
    -(\wt{K}\wt{q} + \wt{F}\wt{p})\rd t
        +
          \wt{N}^\rT
        \rd W,\\
    y
    &=
      \wt{N}\wt{q}
\end{align*}
in terms of (\ref{Knew})--(\ref{Nnew}) and the appropriately transformed positions and momenta
$
    \wt{q}
    :=
    \sqrt{M} q
    $ and
$
    \wt{p}
    :=
    M^{-1/2} p
$
due to $M\succ 0$.  This corresponds to the similarity transformation
\begin{equation}
\label{Asim}
    {\begin{bmatrix}
          \sqrt{M} & 0\\
          0 & M^{-1/2}
    \end{bmatrix}}
    A
    {\begin{bmatrix}
          M^{-1/2}& 0\\
          0 & \sqrt{M}
    \end{bmatrix}}
    =
    {\begin{bmatrix}
          0 & I_n \\
          -\wt{K} & -\wt{F}
    \end{bmatrix}}
\end{equation}
of the matrix $A$ in (\ref{A}).
\end{proof}

In view of (\ref{A})--(\ref{C}) and Lemma~\ref{lem:trans},
the transfer function of the LSH system $(K,M,F,N)$ can be computed as
\begin{align}
\nonumber
    \Phi(s)
    & :=
    C(sI_{2n}-A)^{-1} B\\
\nonumber
    & =
    {\begin{bmatrix}
          \wt{N} & 0
    \end{bmatrix}
    \begin{bmatrix}
          sI_n & -I_n \\
          \wt{K} & sI_n+\wt{F}
    \end{bmatrix}^{-1}
    \begin{bmatrix}
          0 \\
          \wt{N}^{\rT}
    \end{bmatrix}}\\
\label{T}
        & =
    \wt{N} (s^2I_n + s\wt{F} + \wt{K})^{-1} \wt{N}^\rT,
    \qquad
    s \in \mC,
\end{align}
where use is made of the matrix inversion lemma \cite{HJ_2007}. Evaluation of (\ref{T}) at $s=0$ yields the static gain of the system:
\begin{equation}
\label{T0}
    \Phi(0)
    =
    \wt{N} \wt{K}^{-1} \wt{N}^\rT
    =
    NK^{-1}N^\rT,
\end{equation}
assuming that the stiffness matrix $K$ is nonsingular. Note that
the gain $\Phi(0)$ in (\ref{T0}) is a symmetric matrix.
The following theorem provides sufficient conditions for internal stability of the system.

\begin{thm}
\label{th:stab}
Suppose the LSH system in  (\ref{dXlin})--(\ref{C}) has positive definite stiffness and damping matrices $K$ and $F$.
Then the matrix $A$ in (\ref{A}) is Hurwitz.
\hfill$\square$
\end{thm}
\begin{proof}
The fact that the positive definiteness of the matrices $K$, $F$, $M$ ensures the Hurwitz property for $A$ in (\ref{A}) 
can be established by providing
a strict quadratic Lyapunov function specified by the matrix
\begin{equation}
\label{Q}
    Q :=
    R
    +
    \eps
    {\begin{bmatrix}
        0 & 1\\
        1 & 0
    \end{bmatrix}}
    \ox
    I_n
    =
    {\begin{bmatrix}
        K & \eps I_n\\
        \eps I_n & M^{-1}
    \end{bmatrix}},
\end{equation}
where $R$ is the energy matrix from (\ref{R}). Here, $\eps>0$ is a scalar parameter (with $1/\eps$ having the dimension of time) which is small enough in order to guarantee positive definiteness of the matrix $Q$:
\begin{equation}
\label{eps1}
    \eps < \sqrt{\lambda_{\min}(\wt{K})},
\end{equation}
where $\wt{K}$ is given by (\ref{Knew}), and $\lambda_{\min}(\cdot)$ is the smallest eigenvalue (of a matrix with a real spectrum).   The right-hand side of (\ref{eps1}) is the smallest frequency of oscillations in the isolated Hamiltonian system $(K,M,0,0)$ with no damping. The corresponding quadratic form of the state variables for the system $(K,M,F,0)$ (uncoupled from the environment)  is organised as a deformed Hamiltonian $H$ in (\ref{HR}):
\begin{equation}
\label{Ups}
    \Ups
    :=
    \frac{1}{2}
    \|x\|_Q^2
    =
    H + \eps q^\rT p.
\end{equation}
Here, at any point of the phase space $\mR^{2n}$, the additional bilinear term can be represented as
\begin{equation}
\label{inert}
    q^{\rT}p
    =
    q^{\rT} M \dot{q}
    =
    \frac{1}{2}(\|q\|_M^2)^{^\centerdot},
\end{equation}
with $\|q\|_M^2$ resembling the central moment of inertia about the origin (the quantity (\ref{inert}) will play a part in the virial theorem in Section~\ref{sec:inv}).
By a straightforward calculation,
(\ref{A}) and (\ref{Q}) lead to
\begin{equation}
\label{QAAQ}
    \Psi
    :=
    -QA - A^\rT Q
    =
    {\begin{bmatrix}
        2\eps K & \eps FM^{-1}\\
        \eps M^{-1}F & M^{-1} F M^{-1} - 2\eps M^{-1}
    \end{bmatrix}},
\end{equation}
whose right-hand side is
positive definite
if
\begin{equation}
\label{eps2}
    0
    <
    \eps
    <
    \frac{1}{2}
    \lambda_{\min}
    \Big(
    \Big(I_n + \frac{1}{4} \wt{F} \wt{K}^{-1} \wt{F}\Big)^{-1}\wt{F}
    \Big),
\end{equation}
with $\wt{F}$ given by (\ref{Fnew}). Therefore, for any $\eps$ satisfying the constraints (\ref{eps1}) and (\ref{eps2}), the relation (\ref{QAAQ}) implies that $\Ups$ in (\ref{Ups}) is a strict Lyapunov function for the system $(K,M,F,0)$ governed by the ODE $\dot{x} = Ax$, so that $\dot{\Ups} = -\frac{1}{2} \|x\|_{\Psi}^2 < 0$ whenever $x\ne 0$, and hence, $A$ is indeed Hurwitz.
Alternatively, the assertion  of the theorem follows from the invariance of the spectrum of $A$ under the similarity transformation (\ref{Asim}), thus allowing its characteristic polynomial to be computed as
%
\begin{equation}
\label{chi}
    \chi(s)
    :=
    \det(s I_{2n} - A)
     =
    \det
    {\begin{bmatrix}
        s I_n & -I_n\\
        \wt{K} & s I_n + \wt{F}
    \end{bmatrix}}
    =
    \det
    (s^2 I_n + s \wt{F} + \wt{K}),
\end{equation}
which is nonzero for any complex $s := u + iv \in \mC$ with a nonnegative real part $u$ due to positive definiteness of the matrices $\wt{K}$ and $\wt{F}$ in (\ref{Knew}) and (\ref{Fnew}).
Indeed, for any $u\>0$,  the matrix
$
    \wh{F} := \wt{F} + 2u I_n
$
is also positive definite, and hence, for any given $s$ described above, (\ref{chi}) can be represented as
\begin{equation}
\label{chi1}
    \chi(s)
    =
    \det((u^2 - v^2) I_n + \wt{K} + iv \wh{F})
    =
    \det (\wh{K} + iv I_n)
    \det\wh{F},
\end{equation}
where
\begin{equation}
\label{Khat}
    \wh{K}:= \wh{F}^{-1/2}((u^2 - v^2) I_n + \wt{K})\wh{F}^{-1/2}
\end{equation}
is a real symmetric matrix whose spectrum is, therefore, real.
The latter property implies that the eigenvalues of $\wh{K} + iv I_n$  have a common imaginary part $v$ and are all nonzero in the case $v\ne 0$. Furthermore, if $v=0$, then (\ref{Khat}) reduces to
$
    \wh{K} = \wh{F}^{-1/2}(u^2 I_n + \wt{K})\wh{F}^{-1/2}\succ 0
$.
In both cases, the right-hand side of (\ref{chi1}) does not vanish, and hence, the roots of $\chi(s)$ are all in the open left half-plane $\Re s < 0$, whereby $A$ is Hurwitz.
\end{proof}

The strict Lyapunov function $\Ups$ for the isolated damped Hamiltonian system $(K,M,F,0)$ in the proof of Theorem~\ref{th:stab} will be used in dissipation relations of Section~\ref{sec:diss}.

\section{\bf The structure of the invariant measure }\label{sec:inv}

If the linear SDE (\ref{dXlin}) is driven by a standard Wiener process $W$ in $\mR^m$, then its solution $x$ is a Markov diffusion process, which, under the
assumptions of Theorem~\ref{th:stab}, has a unique invariant measure. This measure is  organised as a zero-mean Gaussian probability distribution $\cN(0,\Pi)$ on $\mR^{2n}$ whose covariance matrix
\begin{equation}
\label{Pi}
    \Pi =
    {\begin{bmatrix}
        \Pi_{11} & \Pi_{12}\\
        \Pi_{21} & \Pi_{22}
    \end{bmatrix}}
    =
    \int_0^{+\infty} \re^{tA} BB^{\rT} \re^{tA^\rT} \rd t
\end{equation}
is the controllability Gramian of the pair $(A,B)$, which is a unique solution (due to $A$ being Hurwitz)  of the algebraic Lyapunov equation (ALE)
\begin{equation}
\label{ALE}
    A \Pi + \Pi A^{\rT} + BB^{\rT} = 0.
\end{equation}
In view of (\ref{A}) and (\ref{B}), the blocks $\Pi_{jk} = \Pi_{kj}^\rT\in \mR^{n\x n}$ in (\ref{Pi})
satisfy a set of three algebraic Sylvester equations (ASEs):
\begin{align}
\label{ASE11}
    M^{-1}\Pi_{21} + \Pi_{12}M^{-1}   & = 0,\\
\label{ASE12}
    M^{-1} \Pi_{22} - \Pi_{11}K-\Pi_{12}M^{-1}F & = 0,\\
\label{ASE22}
     -FM^{-1}\Pi_{22}-\Pi_{22}M^{-1}F -K\Pi_{12} - \Pi_{21}K +N^{\rT} N & = 0,
\end{align}
which are the $(1,1)$, $(1,2)$ and $(2,2)$-blocks of the ALE (\ref{ALE}), respectively.
In particular, (\ref{ASE11}) implies that 
\begin{equation}
\label{Pi12Xi}
    \Pi_{12} = \Xi M,
\end{equation}
where $\Xi$ is a real antisymmetric matrix of order $n$.
The covariance structure of the invariant Gaussian measure of the LSH system has a bearing on the following stochastic version of the virial theorem (see, for example, \cite{FBKB_2016}). From (\ref{Pi12Xi}) and the orthogonality of the subspaces of symmetric and antisymmetric matrices, it follows  that, if the system is initialised at the invariant Gaussian distribution $\cN(0,\Pi)$, then
\begin{equation}
\label{Eqp}
    \bE(q^\rT p)
    =
    \Tr \Pi_{12}
    =
    \bra
        \Xi, M
    \ket
    =
    0.
\end{equation}
Alternatively, since the process $q$ has continuously differentiable sample paths, then so also does $\|q\|_M^2$. This makes the identities (\ref{inert})
applicable, and hence,
$
    \bE(q^\rT p)
    =
    \frac{1}{2}
    (\bE (\|q\|_M^2))^{^\centerdot}
    =
    0
$,
in accordance with (\ref{Eqp}). Although the sample paths of $p$ are not differentiable (moreover, have infinite variation and finite quadratic variation due to (\ref{Wquad})), application of the Ito lemma  shows that the smoothness of $q$ makes the Ito correction term disappear in
\begin{align}
\nonumber
    \rd (q^\rT p)
    & =
    (\rd q)^\rT p
    +
    q^\rT \rd p
    +
    (\rd q)^\rT \rd p\\
\nonumber
    & =
    \dot{q}^\rT p\rd t
    +
    q^\rT \rd p\\
\label{dqp}
    &=
    (\dot{q}^\rT p + q^\rT f)\rd t
    +
    y^\rT
    \rd W,
\end{align}
where use is made of (\ref{qdotlin})--(\ref{ylin}) along with the total internal (restoring and damping) force
\begin{equation}
\label{f}
    f:= -Kq - FM^{-1}p.
\end{equation}
  The averaging of both sides of the SDE (\ref{dqp}) over the invariant measure (with the martingale term $    y^\rT
    \rd W
$ making no contribution) leads to 
\begin{equation}
\label{Edqp0}
    \bE(\dot{q}^\rT p + q^\rT f)=0.
\end{equation}
Since
\begin{equation}
\label{2T}
    \dot{q}^\rT p = \|p\|_{M^{-1}}^2 = 2T
\end{equation}
is twice the kinetic energy of the system in (\ref{kin}), then (\ref{Edqp0}) is indeed equivalent to the virial theorem:
\begin{equation}
\label{vir}
    \bE T = -\frac{1}{2}\bE (q^\rT f).
\end{equation}
The relation (\ref{Edqp0}) (or its equivalent form (\ref{vir})) corresponds to taking the trace on both sides of (\ref{ASE12}) and using (\ref{f}).
An additional insight into the invariant measure of the LSH system employs stochastic filtering structures \cite{LS_2001} and comes from
the fact that the right-hand side of the ODE (\ref{qdotlin}) is adapted to the natural filtration $\cP:= (\cP_t)_{t\>0}$ of the momentum process $p$.
Suppose the LSH system is initialised at the invariant distribution, so that the random vector $x(0)$ is $\cN(0,\Pi)$-distributed  and independent of the Wiener process $W$. Then the SDE (\ref{dplin}) can be represented as
\begin{equation}
\label{dpe}
    \rd p
    =
    (\wh{f}-Ke)   \rd t
        +
          N^\rT
        \rd W,
\end{equation}
where
\begin{align}
\label{fhat}
    \wh{f}(t)
    & := \bE(f(t)\mid \cP_t) = -K\wh{q}(t) - FM^{-1}p(t),\\
\label{qhat}
    \wh{q}(t)
    & :=
    \bE(q(t)\mid \cP_t)
    =
    \bE(q(0)\mid \cP_t) + M^{-1}\int_0^t p(\tau)\rd \tau
\end{align}
are the conditional expectations of the current internal force (\ref{f}) and the position with respect to the
$\sigma$-algebra $\cP_t$ (generated by the past history of the momentum process $p$ over the time interval $[0,t]$), 
with
\begin{equation}
\label{qq}
    e:= q-\wh{q}
\end{equation}
the corresponding ``estimation'' error.
Note that (\ref{qhat}) reduces to estimating the initial position $q(0)$ since the
integral part of the solution $q(t) = q(0) + M^{-1} \int_0^t p(\tau)\rd \tau$ of
the ODE (\ref{qdotlin}) is $\cP$-adapted.
The following lemma will allow the calculation of the filtering estimates to avoid the Moore-Penrose pseudoinverse \cite{HJ_2007}.

\begin{lem}
\label{lem:contr}
Suppose the LSH system (\ref{dXlin})--(\ref{C}) satisfies the conditions of Theorem~\ref{th:stab} and its  coupling matrix $N$ is of full column rank:
\begin{equation}
\label{NNpos}
    D:= N^\rT N\succ 0.
\end{equation}
Then the invariant covariance matrix $\Pi$ in (\ref{Pi}) is nonsingular. \hfill $\square$
\end{lem}
\begin{proof}
From (\ref{A}) and (\ref{B}), it follows that the Kalman controllability matrix
$
    \Lambda:= \begin{bmatrix}B&  AB&  \ldots&  A^{2n-1}B\end{bmatrix}
$
satisfies
\begin{equation}
\label{contr}
    \Lambda \Lambda^\rT
    \succcurlyeq
    BB^\rT
    +ABB^\rT A^\rT
    =
    {\begin{bmatrix}
        M^{-1} D M^{-1} & -M^{-1} D M^{-1}F\\
        -FM^{-1} D M^{-1} & D + FM^{-1} D M^{-1}F
    \end{bmatrix}}
    \succ 0
\end{equation}
in view of the positive definiteness of $M$ and $D$ (the Schur complement of the block $M^{-1} D M^{-1}\succ 0$ on the right-hand side of (\ref{contr}) is $D\succ 0$ due to (\ref{NNpos})). Hence, the pair $(A,B)$ is controllable, which is equivalent to $\Pi \succ 0$ in  (\ref{Pi}).
\end{proof}

Since $q$ and $p$ are jointly Gaussian random processes, application of the Kalman stochastic filtering theory \cite{LS_2001} under the conditions of Lemma~\ref{lem:contr} shows that the estimate (\ref{qhat}) satisfies the SDE
\begin{equation}
\label{dqhat}
    \rd \wh{q} = M^{-1}p \rd t - P K D^{-1}(\rd p - \wh{f}\rd t )
\end{equation}
(driven by an innovation process with respect to $\cP$ whose stochastic differential $\rd p - \wh{f}\rd t$ involves (\ref{fhat})),
with the initial condition
\begin{equation}
\label{qhat0}
    \wh{q}(0) = \bE(q(0)\mid \cP_0) = \Pi_{12} \Pi_{22}^{-1} p(0).
\end{equation}
Here,
\begin{equation}
\label{P}
    P(t):= \cov(q(t)\mid \cP_t) = \cov(q(0)\mid \cP_t) = \cov(e(t))
\end{equation}
is the covariance matrix of the Gaussian  conditional distribution of $q(t)$ with respect to $\cP_t$ for any $t\>0$ (which coincides with the unconditional covariance matrix of the position filtering error (\ref{qq}) in the Gaussian case). Due to the $\cP$-adaptedness of the right-hand side of (\ref{qdotlin}) mentioned above,  the Riccati ODE for $P$ in (\ref{P}) reduces to
\begin{equation}
\label{Pdot}
  \dot{P} = - P K D^{-1} K P,
\end{equation}
with
\begin{equation}
\label{P0}
    P(0) = \cov(q(0)\mid \cP_0) = \Pi_{11}-\Pi_{12}\Pi_{22}^{-1} \Pi_{21} \succ 0
\end{equation}
in accordance with (\ref{qhat0}) and the property that $\Pi\succ 0$. Since (\ref{Pdot}) is equivalent to $ (P^{-1})^{^\centerdot} = K D^{-1} K$, then its solution is given by
\begin{equation}
\label{Psol}
  P(t) = (P(0)^{-1} + tK D^{-1} K)^{-1},
  \qquad
  t\>0.
\end{equation}
In view of (\ref{dpe}) and (\ref{dqhat}), the estimation error in (\ref{qq}) satisfies the SDE
$
    \rd e = PKD^{-1}(-Ke \rd t + N^\rT \rd W)
$,
with $e(0)$ being $\cN(0,P(0))$-distributed with the covariance matrix (\ref{P0}) and independent of the initial momentum $p(0)$ and the Wiener process $W$. The asymptotic behaviour $P(t)\sim \frac{1}{t} K^{-1} D K^{-1}$ of the covariance matrix (\ref{Psol}), as $t \to +\infty$, reflects the accumulation  of information in the momentum process $p$ about the initial position $q(0)$ of the LSH system.

\section{\bf Robust stochastic stability
 }\label{sec:diss}

We will now consider dissipation relations for the quadratic function $\Ups$ of the system variables in (\ref{Ups}) in the case when $W$ is a statistically uncertain random force in the form of an Ito process (\ref{dW}) instead of the standard Wiener process. Assuming that the parameter  $\eps$ satisfies the conditions (\ref{eps1}) and (\ref{eps2}) of the proof of Theorem~\ref{th:stab}, a combination of (\ref{dH}), (\ref{dqp}), (\ref{2T}) leads to
\begin{align}
\nonumber
    \rd \Ups
    = &
    \rd H + \eps \rd (q^\rT p)\\
\nonumber
        = &
        \Big(
            -
            \|\dot{q}\|_F^2
    +
            \frac{1}{2}
            \bra
                N M^{-1}N^\rT,
                \Sigma
            \ket
        \Big)\rd t
        +
        \dot{q}^\rT
        N^\rT \rd W\\
\nonumber
        & +
    \eps((\dot{q}^\rT p + q^\rT f)\rd t
 +
    y^\rT
    \rd W)\\
\nonumber
    =&
    \frac{1}{2}
    (-\|x\|_{\Psi}^2 +             \bra
                \wt{N} \wt{N}^\rT,
                \Sigma
            \ket)\rd t
            +(\eps q + M^{-1}p)^\rT N^\rT \rd W\\
\label{dUps}
    =&
    \frac{1}{2}
    (-\|x\|_{\Psi}^2 +
    \bra
                \wt{N} \wt{N}^\rT,
                \Sigma
            \ket
            +
            2x^\rT \Gamma \alpha
            )\rd t
            +x^\rT
            \Gamma \beta \rd \omega.
\end{align}
Here, $\alpha$ and $\Sigma$ are the drift vector and the diffusion matrix of $W$ from (\ref{dW}) and (\ref{diffmat}). Also, use is made of the matrices $\wt{N}$ and  $\Psi$ from (\ref{Nnew}) and (\ref{QAAQ}) together with an auxiliary matrix $\Gamma\in \mR^{2n\x m}$ given by
\begin{equation}
\label{Gamma}
  \Gamma
  :=
            \begin{bmatrix}
                \eps I_n\\
                M^{-1}
            \end{bmatrix}
            N^\rT.
\end{equation}
Now, if both $\|\Sigma\|$ and $|\alpha|^2$ are almost surely bounded from above by quadratic functions of the system variables (with constant coefficients), then, in view of the Cauchy-Bunyakovsky-Schwarz inequality,
\begin{equation}
\label{class}
    \bra
                \wt{N} \wt{N}^\rT,
                \Sigma
            \ket
            +
            2x^\rT \Gamma \alpha
            \< \gamma + \|x\|_{\Delta}^2
\end{equation}
holds for a scalar $\gamma\>0$ and a real positive semi-definite symmetric matrix $\Delta$ of order $2n$. The inequality (\ref{class}) can be based on prior information about the Ito process $W$ 
and, in fact, describes a particular class of statistical uncertainties in the external random force.

\begin{thm}
\label{th:diss} Suppose the LSH system (\ref{dXlin})--(\ref{C}) satisfies the conditions of Theorem~\ref{th:stab}, and the parameter $\eps$ satisfies (\ref{eps1}) and (\ref{eps2}). Also, suppose the initial system variables have finite second moments (that is, $\bE(|x(0)|^2) <+\infty$), and the uncertainty class for the random external force (\ref{dW}) is described by (\ref{Gamma}) and (\ref{class}) with a sufficiently ``small'' matrix $\Delta$ compared to the matrix $\Psi$ in (\ref{QAAQ}):
\begin{equation}
\label{DeltaPsi}
  \Delta \prec \Psi.
\end{equation}
Then the subsequent second moments of the system variables satisfy
\begin{equation}
\label{diss}
    \limsup_{t\to +\infty}
    \bE(|x(t)|^2)
    \< \frac{\gamma}{\lambda_{\min}(Q)\mu},
\end{equation}
where
\begin{equation}
\label{mu}
    \mu:= \lambda_{\min}((\Psi-\Delta)Q^{-1}),
\end{equation}
and the matrix $Q$ is given by (\ref{Q}).\hfill$\square$
\end{thm}
\begin{proof}
The fulfillment of (\ref{eps1}) and (\ref{eps2}) ensures that $Q\succ 0$ and $\Psi\succ 0$, while (\ref{DeltaPsi}) yields $\mu>0$. Now, the random process
\begin{equation}
\label{Z}
    Z(t):= \re^{\mu t} \Big(\Ups(t) - \frac{\gamma}{2\mu}\Big),
    \qquad
    t\>0,
\end{equation}
is a supermartingale \cite{RW_2000} with respect to the filtration $\cF$, which, in view of (\ref{dUps})--(\ref{mu}), satisfies an SDE with a non-positive drift:
\begin{align}
\nonumber
    \rd Z
    &=
    \mu Z\rd t + \re^{\mu t} \rd \Ups\\
\nonumber
    &=
    \re^{\mu t}
    \Big(\Big(\mu \Ups  -\frac{\gamma}{2}\Big)\rd t+ \rd \Ups \Big)\\
\nonumber
    &=
    \frac{1}{2}
    \big(-\|x\|_{\Psi-\mu Q}^2 +
    \bra
                \wt{N} \wt{N}^\rT,
                \Sigma
            \ket
            +
            2x^\rT \Gamma \alpha
            -\gamma
            \big)\rd t
            +x^\rT
            \Gamma \beta \rd \omega.
\end{align}
Hence, $\bE Z(t)$ is a nonincreasing function of time $t\>0$, which, in combination with (\ref{Z}), implies that
\begin{equation}
\label{EUps}
    \bE \Ups(t)
    \<
    \frac{\gamma}{2\mu}
    +
    \re^{-\mu t}
    \Big(
        \bE \Ups(0)
        -
        \frac{\gamma}{2\mu}
    \Big),
\end{equation}
where $\bE \Ups(0) \< \frac{1}{2} \lambda_{\max}(Q) \bE(|x(0)|^2)<+\infty$ due to (\ref{Ups}) and the assumption of the theorem.  The upper bound (\ref{diss}) can now be obtained from (\ref{EUps}) and the inequality $\bE \Ups \> \frac{1}{2} \lambda_{\min}(Q) \bE(|x|^2)$.
\end{proof}

\section{\bf Feedback connection of LSH systems}\label{sec:loop}

Consider a feedback connection of two LSH systems $S_k:= (K_k,M_k,F_k,N_k)$, $k=1,2$ (interpreted as a plant and a controller, respectively), which have $n$ degrees of freedom and are driven by 
$\mR^m$-valued Ito processes $W_1$ and $W_2$ (with respect to a common filtration $\cF$ such that the initial system  states are $\cF_0$-measurable), see  Fig.~\ref{fig:loop}.
\begin{figure}[htpb]
\centering
\unitlength=1mm
\linethickness{0.2pt}
\begin{picture}(50.00,43.00)
    \put(20,30){\framebox(10,10)[cc]{{$S_1$}}}
    \put(20,10){\framebox(10,10)[cc]{{$S_2$}}}
    \put(40,15){\vector(0,1){17}}
    \put(30,15){\line(1,0){10}}

    \put(20,35){\line(-1,0){10}}
    \put(10,35){\vector(0,-1){17}}
    \put(10,15){\circle{6}}
    \put(10,15){\makebox(0,0)[cc]{$+$}}
    \put(13,15){\vector(1,0){7}}
    \put(-3,15){\vector(1,0){10}}
    \put(-8,15){\makebox(0,0)[cc]{{$W_2$}}}
\put(42,15){\makebox(0,0)[lc]{{$y_2$}}}

    \put(37,35){\vector(-1,0){7}}
    \put(40,35){\circle{6}}
    \put(40,35){\makebox(0,0)[cc]{$+$}}
    \put(53,35){\vector(-1,0){10}}
    \put(55,35){\makebox(0,0)[lc]{{$W_1$}}}
    \put(8,35){\makebox(0,0)[rc]{{$y_1$}}}
\end{picture}\vskip-8mm
\caption{The feedback connection of two LSH systems described by (\ref{qdotlink})--(\ref{ylink}).
}
\label{fig:loop}
\end{figure}
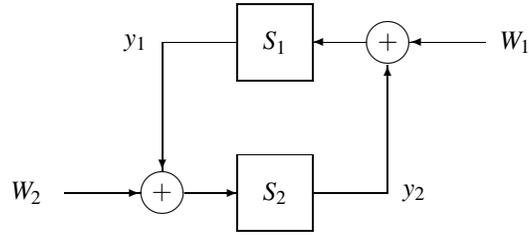
Similarly to the deterministic counterpart \cite[Section II]{V_2016} of the stochastic setting,  this system connection is governed by 
\begin{align}
\label{qdotlink}
    \dot{q}_k
    & =
    M_k^{-1} p_k,\\
\nonumber
    \rd p_k
    & =
    -(K_kq_k + F_kM_k^{-1}p_k)\rd t
        +
          N_k^\rT
        (y_{3-k}\rd t + \rd W_k)\\
\label{dplink}
    & =
    -(K_kq_k-N_k^\rT N_{3-k}q_{3-k} + F_kM_k^{-1}p_k)\rd t
        +
          N_k^\rT
        \rd W_k,\\
\label{ylink}
    y_k &
    =
      N_kq_k
\end{align}
for $k=1,2$. Here, with a slight abuse of notation, $q_1$, $q_2$ and $p_1$, $p_2$ are the corresponding $n$-dimensional position and momentum vectors 
which are assembled into a $4n$-dimensional state vector $x$ as
\begin{equation}
\label{xqp}
    x
    :=
    {\begin{bmatrix}
        q\\
        p
    \end{bmatrix}},
    \qquad
    q
    :=
    {\begin{bmatrix}
        q_1\\
        q_2
    \end{bmatrix}},
    \qquad
    p
    :=
    {\begin{bmatrix}
        p_1\\
        p_2
    \end{bmatrix}}.
\end{equation}
The closed-loop system is also an LSH system whose position and momentum vectors $q$ and $p$ are
driven by the $2m$-dimensional random process
\begin{equation}
\label{WW}
    W
    :=
    {\begin{bmatrix}
        W_1\\
        W_2
    \end{bmatrix}}.
\end{equation}
The quadruple $(K,M,F,N)$ of this system is computed as
\begin{align}
\label{KM}
    K
    & =
    {\begin{bmatrix}
        K_1 & -N_1^\rT N_2\\
        -N_2^\rT N_1 & K_2
    \end{bmatrix}},
    \qquad
    M
    =
    {\begin{bmatrix}
        M_1 & 0\\
        0 & M_2
    \end{bmatrix}},\\
\label{FN}
    F
    & =
    {\begin{bmatrix}
        F_1 & 0\\
        0 & F_2
    \end{bmatrix}},
    \qquad
    \qquad
    \qquad\
    N
    =
    {\begin{bmatrix}
        N_1 & 0\\
        0 & N_2
    \end{bmatrix}},
\end{align}
cf. \cite[Eq. (29)]{V_2016}. The mass and damping matrices $M$ and $F$  inherit positive definiteness from the corresponding matrices of the subsystems $S_1$ and $S_2$.  However, for given $K_1\succ 0$ and $K_2 \succ 0$, the stiffness matrix $K$ in (\ref{KM}) is positive definite if and only if $K_1-N_1^\rT N_2 K_2^{-1}N_2^\rT N_1\succ 0$, which is equivalent to  the coupling matrices in (\ref{FN}) being small enough in the sense that
\begin{equation}
\label{KNNK}
    \|K_1^{-1/2} N_1^\rT N_2 K_2^{-1/2}\| <1,
\end{equation}
where $\|\cdot\|$ is the operator norm of a matrix. In accordance with \cite{LP_2008}, the condition (\ref{KNNK}) can be expressed in terms of the static gains (\ref{T0}) of the systems $S_1$ and $S_2$ due to the relation
$
    \|K_1^{-1/2} N_1^\rT N_2 K_2^{-1/2}\|^2
     =
    \lambda_{\max}
    (\Phi_1(0) \Phi_2(0))
$,
where $\lambda_{\max}(\cdot)$ is the largest eigenvalue (of a matrix with a real spectrum).
Since the operator norm is submultiplicative, a sufficient condition for (\ref{KNNK}) (in the spirit of the small-gain theorem)  is
$
    \|\Phi_1(0)\|\|\Phi_2(0)\| <1
$. In combination with (\ref{KNNK}), the robust stability estimates of Section~\ref{sec:diss} can be applied to a suboptimal choice
of an LSH controller $S_2$ so as to achieve guaranteed upper bounds on the second-order moments of the closed-loop system variables (\ref{xqp}) in the presence of statistical uncertainty in the random process (\ref{WW}).
\section{\bf Conclusion}\label{sec:conc}

We have briefly discussed some of the dynamic and probabilistic aspects of stochastic Hamiltonian systems driven by random forces. In particular, we have considered stability conditions,  energy balance relations, the structure  of the invariant measure and stochastic robust stability  for LSH systems in the presence of statistical uncertainty. These results are applicable to robust control of LSH systems, which is
 intended for future publications.

%

%

\end{document}